\documentclass[11pt,twoside]{amsart}
\usepackage{latexsym,amssymb,amsmath,curves}

\numberwithin{equation}{section}
\newtheorem{theorem}{Theorem}[section]

\newtheorem{corollary}[theorem]{Corollary}

\theoremstyle{definition}

\newtheorem{remark}[theorem]{Remark}

\begin{document}


\title[Projective Parameterized Linear Codes ]{Projective parameterized linear codes arising from some matrices and their main parameters} 

\author[Manuel G. Sarabia]{Manuel Gonz\'alez Sarabia}
\address{
Departamento de Ciencias B\'asicas\\
Unidad Profesional Interdisciplinaria
en Ingenier\'{\i}a y Tecnolog\'{\i}as Avanzadas\\ 
Instituto Polit\'ecnico Nacional\\ 
07340, M\'exico, D.F.
}\email{mgonzalezsa@ipn.mx}

\author[Carlos Renter\'{\i}a M.]{Carlos Renter\'{\i}a M\'arquez}
\address{
Departamento de Matem\'aticas\\
Escuela Superior de F\'{\i}sica y Matem\'aticas\\ 
Instituto Polit\'ecnico Nacional\\ 
07300, M\'exico, D.F.
}\email{renteri@esfm.ipn.mx}

\author[Eliseo Sarmiento R.]{Eliseo Sarmiento Rosales}
\address{
Departamento de
Matem\'aticas\\
Centro de Investigaci\'on y de Estudios
Avanzados del
IPN\\
Apartado Postal
14--740 \\
07000, M\'exico, D.F.
}
\email{esarmiento@math.cinvestav.mx}

\thanks{The first two authors are partially supported by COFAA--IPN and SNI--SEP. 
The third author is partially supported by CONACyT}

 
\subjclass[2000]{Primary 13P25; Secondary 14G50, 14G15, 11T71, 94B27, 94B05.}

\begin{abstract}
In this paper we will estimate the main parameters of some eva\-lua\-tion codes which are known as projective parameterized codes. We will find the length of these codes and we will give a formula for the dimension in terms of the Hilbert function associated to two ideals, one of them being the vanishing ideal of the projective torus. Also we will find an upper bound for the minimum distance and, in some cases, we will give some lower bounds for the regularity index and the minimum distance. These lower bounds work in several cases, particularly for any projective parameterized code associated to the incidence matrix of uniform clutters and then they work in the case of graphs.
\end{abstract}

\maketitle

\section{Introduction}
Let $K=\mathbb{F}_q$  be a finite field with $q$ elements and $L=K[Z_1,\ldots,Z_n]$ be a polynomial ring over the field $K$.
Let $Z^{a_1},\ldots,Z^{a_m}$ be a finite set of monomials.  
As usual if $a_i=(a_{i1},\ldots,a_{in})\in\mathbb{N}^n$, where $\mathbb{N}$ stands for the non-negative integers, 
then we set
 
\begin{center}
$Z^{a_i}=Z_1^{a_{i1}}\cdots Z_n^{a_{in}} \, \, \mbox{for all} \, \, i=1,\ldots,m$.
\end{center}

Consider the following set
parameterized  by these mo\-no\-mials 
\begin{equation} \label{toricset}
X=\left\{\left[\left(t_1^{a_{11}}\cdots t_n^{a_{1n}},\ldots,t_1^{a_{m1}}\cdots t_n^{a_{mn}}\right)\right]\in\mathbb{P}^{m-1}	\vert\, t_i\in K^*\right\},
\end{equation}
where $K^*=K\setminus\{0\}$ and $\mathbb{P}^{m-1}$ is a projective
space over the field $K$. 
Following \cite{afinetv} we call $X$ an 
algebraic toric set parameterized  by
$Z^{a_1},\ldots,Z^{a_m}$. The set $X$ is a multiplicative group under
componentwise multiplication. 

In the same way, let $A$ be the $n \times m$ matrix given by
\begin{equation}
\begin{pmatrix} \label{matrix}
a_{11} & a_{21} & \cdots & a_{m1} \\
a_{12} & a_{22} & \cdots & a_{m2} \\
\vdots & \vdots & \vdots & \vdots \\
a_{1n} & a_{2n} & \cdots & a_{mn}
\end{pmatrix}.
\end{equation}
We say that the set defined in (\ref{toricset}) is the algebraic toric set associated to the matrix $A$. We note that 
$$[(t_1^{a_{11}}\cdots t_{n}^{a_{1n}},t_1^{a_{21}}\cdots t_n^{a_{2n}},\ldots, t_1^{a_{m1}}\cdots t_n^{a_{mn}})]=
[(1,t_1^{a_{21}-a_{11}}\cdots t_n^{a_{2n}-a_{1n}},\ldots,t_1^{a_{m1}-a_{11}}\cdots t_n^{a_{mn}-a_{1n}})].$$
By taking $b_{ij}=a_{ij}-a_{1j}$ for all $i=2,\ldots,m$, $j=1,\ldots,n$, we obtain 
\begin{equation} \label{realtoricset}
X=\left\{\left[\left(1,t_1^{b_{21}}\cdots t_n^{b_{2n}},\ldots,t_1^{b_{m1}}\cdots t_n^{b_{mn}}\right)\right] \in {\mathbb{P}}^{m-1}: t_i \in K^*\right\}.
\end{equation}

From now on we will use any of the representations (\ref{toricset}) or (\ref{realtoricset}) to mean the algebraic toric set parameterized by the monomials $Z^{a_1},\ldots,Z^{a_{m}}$ or, in an equivalent way, to represent the algebraic toric set associated to the matrix (\ref{matrix}).

Let
$S=K[X_1,\ldots,X_m]=\oplus_{d=0}^\infty S_d$ 
be a polynomial ring 
over the field $K$ with the standard grading, let $[P_1],\ldots,[P_{|X|}]$
be the points of $X$, and 
let $f_0(X_1,\ldots,X_m)=X_1^d$. The evaluation map
\setlength{\arraycolsep}{0.0em} 
\begin{eqnarray} \label{ev-map}
& {\rm ev}_d\colon S_d=K[X_1,\ldots,X_m]_d\rightarrow K^{|X|},\nonumber\\
& f\mapsto \left(\frac{f(P_1)}{f_0(P_1)},\ldots,\frac{f(P_{|X|})}{f_0(P_{|X|})}\right)
\end{eqnarray}
\setlength{\arraycolsep}{5pt}
\noindent defines a linear map of
$K$-vector spaces. The image of ${\rm ev}_d$, denoted by $C_X(d)$,
defines a linear code. We will call
$C_X(d)$ a projective parameterized code of
order $d$ arising from the toric set $X$ or associated to the matrix $A$. As usual by a linear code we mean a linear subspace of
$K^{|X|}$. 

In this paper we will only
deal with projective parameterized codes arising from the set $X$, defined in (\ref{toricset}) or (\ref{realtoricset}), over finite fields and we will describe their main characteristics.

The dimension and length of the code $C_X(d)$ 
are given by $\dim_K C_X(d)$ and $|X|$ respectively. The dimension
and length 
are two of the basic parameters of a linear code. A third
basic parameter is the minimum
distance which is given by 
$$\delta_X(d)=\min\{\|v\|
\colon 0\neq v\in C_X(d)\},$$ 
where $\|v\|$ is the number of non-zero
entries of $v$. The basic parameters of $C_X(d)$ are related by the
Singleton bound which is an upper bound for the minimum distance
$$
\delta_X(d)\leq |X|-\dim_KC_X(d)+1.
$$
Projective parameterized codes are important because in some cases their main parameters have the best behavior. For example in \cite{GRS} the resulting  codes are MDS.

The parameters of evaluation codes over finite fields have been computed in several cases. Our approximation, when we consider the evaluation codes as associated to the matrix (\ref{matrix}), generalizes many cases studied previously. For example if $A=I_m$, the projective parameterized codes associated to $A$ become the Generalized Reed-Solomon codes \cite{GRH}. If $X=\mathbb{P}^{m-1}$,  
the parameters of $C_X(d)$ are described in
\cite[Theorem~1]{sorensen}. If $X$ is the image of the affine space
$\mathbb{A}^{m-1}$ under the map $\mathbb{A}^{m-1}\rightarrow
\mathbb{P}^{m-1} $, $x\mapsto [(1,x)]$, the parameters 
of $C_X(d)$ are described in
\cite[Theorem~2.6.2]{delsarte-goethals-macwilliams}. Also if we consider the matrix (\ref{matrix}) as the incidence matrix of a graph $G$, we obtain the projective parameterized codes associated to $G$. In the following sections when we write graph we mean a simple graph, i.e., an undirected graph that has no loops and no more than one edge between any two different vertices. The main characteristics of evaluation codes associated to complete bipartite graphs were found in \cite{GR}. Some general results over projective parameterized codes were described in \cite{algcodes}. 

It is worth saying that projective parameterized codes are, in general, strictly different to toric codes which were defined in \cite{hansen} and generalized for example in \cite{joyner} and \cite{soprunov}. They evaluate over the complete torus, meanwhile we do it over specific subsets of the projective space.

In this work we will analyze the case where the pa\-ra\-me\-te\-ri\-zed codes of order $d$, $C_X(d)$, come from the general matrix (\ref{matrix}) and we will estimate their main parameters.

The vanishing ideal of $X$, denoted by $I_X$, is the
ideal of $S$ generated by the homogeneous polynomials of $S$ that
vanish on $X$. 

For all unexplained
terminology and additional information  we refer to
\cite{AL,Sta1} (for the theory of polynomial ideals and Hilbert functions), and
\cite{MacWilliams-Sloane,stichtenoth,tsfasman} (for the theory of 
error-correcting codes and algebraic geometric codes).

\section{Preliminaries} \label{sec2}
We continue using the notation and definitions given in the
introduction. In this section we introduce the basic algebraic
invariants of $S/I_X$ and their co\-nnec\-tion with the basic 
parameters of projective parameterized linear codes. Then we present some of the
results that we are going to use later. 

Recall that the projective space of 
dimension $m-1$ over $K$, denoted by 
$\mathbb{P}^{m-1}$, is the quotient space 
$$(K^{m}\setminus\{0\})/\sim $$
where two points $\alpha$, $\beta$ in $K^{m}\setminus\{0\}$ 
are equivalent if $\alpha=\lambda{\beta}$ for some $\lambda\in K$. We
denote the  
equivalence class of $\alpha$ by $[\alpha]$. Let
$X\subset\mathbb{P}^{m-1}$ be an algebraic toric set
parameterized by $Z^{a_1},\ldots,Z^{a_m}$ and let $C_X(d)$ be a
projective parameterized code of order $d$. The kernel of the
evaluation map ${\rm ev}_d$, defined in 
Eq.~(\ref{ev-map}), is precisely $I_X(d)$ the degree
$d$ piece of $I_X$. Therefore there is an isomorphism of $K$-vector spaces
$$S_d/I_X(d)\simeq C_X(d).$$

Two of the basic parameters of $C_X(d)$ can be expressed 
using Hilbert functions of standard graded algebras \cite{Sta1}, as we
now  explain. Recall that the
Hilbert function of
$S/I_X$ is given by 

\begin{center}
$H_X(d)=\dim_K 
(S/I_X)_d=\dim_K 
S_d/I_X(d)=\dim_KC_X(d)$.
\end{center}

The unique polynomial $h_X(t)=\sum_{i=0}^{k-1}c_it^i\in
\mathbb{Z}[t]$ of degree $k-1=\dim(S/I_X)-1$ such that $h_X(d)=H_X(d)$ for
$d\gg 0$ is called the Hilbert polynomial of $S/I_X$. The
integer $c_{k-1}(k-1)!$, denoted by ${\rm deg}(S/I_X)$, is 
called the degree or multiplicity of $S/I_X$. In our
situation 
$h_X(t)$ is a non-zero constant because $S/I_X$ has dimension $1$. 
Furthermore $h_X(d)=|X|$ for $d\geq |X|-1$, see \cite[Lecture
13]{harris}. This means that $|X|$ equals the degree 
of $S/I_X$. Thus $H_X(d)$ and ${\rm deg}(S/I_X)$ equal the
dimension and the length of $C_X(d)$ res\-pec\-ti\-ve\-ly. There are algebraic
methods, based on elimination theory and Gr\"obner bases, to compute
the dimension and the length of $C_X(d)$ \cite{algcodes}. 

The regularity index of $S/I_X$, denoted by 
${\rm reg}(S/I_X)$, is the least integer $p\geq 0$ such that
$h_X(d)=H_X(d)$ for $d\geq p$. The degree and the regularity index can be
read off the Hilbert series as we now explain. The Hilbert series of
$S/I_X$ can be 
written as
$$
F_X(t)=\sum_{i=0}^{\infty}H_X(i)t^i=\sum_{i=0}^{\infty}\dim_K(S/I_X)_it^i=
\frac{h_0+h_1t+\cdots+h_rt^r}{1-t},
$$
where $h_0,\ldots,h_r$ are positive integers. In fact we have that 
$h_i=\dim_K(S/(I_X,X_m))_i$ for $0\leq i\leq r$ and
$\dim_K(S/(I_X,X_m))_i=0$ for $i>r$.  
This follows from the fact that $I_X$ is a Cohen-Macaulay lattice
ideal \cite{algcodes} and by observing that $\{X_m\}$ is a regular system of
parameters for $S/I_X$ (see \cite{Sta1}). The number 
$r$  equals the regularity index of $S/I_X$ and the degree of
$S/I_X$ equals $h_0+\cdots+h_r$ (see \cite{Sta1} or 
\cite[Corollary~4.1.12]{monalg}). 

The regularity index plays a very important role in the study of evaluation codes arising from a set $X$ because in the case $d \geq {\rm reg} \, (S/I_X)$ we obtain that $H_X(d)=|X|$ and then $C_X(d)=K^{|X|}$, which is a trivial case. Therefore we always work with $0 \leq d < {\rm reg} \, (S/I_X)$. Another motivation to study the regularity index comes from commutative algebra because, in this case, ${\rm reg} \, (S/I_X)$ is equal to the Castelnuovo--Mumford regularity which is an algebraic invariant of central importance \cite{eisenbud}.

\section{Main Results} From now on we will work with the toric set $X$ defined in Eqs. (\ref{toricset}) or (\ref{realtoricset}) and our goal is to describe the main parameters of the projective parameterized codes of order $d$, $C_X(d)$, which were defined as the image of the evaluation map ${\rm ev}_d$ introduced in Eq. (\ref{ev-map}).

\subsection{Length} \label{length1}
In order to cumpute the length of the projective parameterized codes arising from the toric set $X$, we introduce the following multiplicative subgroups of $X$ for all $i=1,\ldots,n$.
$$
Y_i:=\{[(1,t_i^{b_{2i}},\ldots,t_i^{b_{mi}})] \in \mathbb{P}^{m-1}: t_i \in K^* \, \, {\mbox{for all}}\, \, i\}.
$$

It is easy to see that $|Y_i|=\frac{q-1}{(q-1,b_{2i},\ldots,b_{mi})}$ for all $i=1,\ldots,n$ and where $(q-1,b_{2i},\ldots,b_{mi})$ means the greatest common divisor of the corresponding integers. With this information we are able to prove the main result of this section.

\begin{theorem} \label{theorem-length}
The length of the projective parameterized codes of order $d$, $C_X(d)$, is given by
\begin{equation} \label{length}
|X|=\frac{1}{|M|} \prod_{i=1}^n |Y_i|
\end{equation}

\noindent where $M$ is the set of $n$-tuples $(i_1,\ldots,i_{n})$ such that

\begin{center}
$1 \leq i_j \leq \frac{q-1}{(q-1,b_{2j},\ldots,b_{mj})}$ \,\, for all \, $j=1,\ldots,n$,
\end{center}

\noindent and
\setlength{\arraycolsep}{0.0em}
\begin{eqnarray} \label{system}
& i_1b_{21}+ i_2b_{22}+ \cdots + i_nb_{2n} \equiv 0 \, {\rm mod} \, (q-1) \nonumber\\
& i_1b_{31}+ i_2b_{32}+ \cdots + i_nb_{3n} \equiv 0 \, {\rm mod} \, (q-1) \nonumber\\
& \cdots \, \cdots \, \cdots \, \cdots \,\cdots  \cdots \, \cdots \, \cdots \, \cdots \, \cdots \, \cdots \, \cdots \\
& i_1b_{m1} + i_2b_{m2}+\cdots + i_nb_{mn} \equiv 0 \, {\rm mod} \, (q-1) \nonumber
\end{eqnarray}
\setlength{\arraycolsep}{5pt}
\end{theorem}
\begin{proof}
Let $\phi$ be the following map
\begin{center} \label{phi}
$\phi: Y_1 \times \cdots \times Y_n \rightarrow X,$
\end{center}

\begin{center}
$\phi([(1,t_1^{b_{21}},\ldots,t_1^{b_{m1}})],\ldots,[(1,t_n^{b_{2n}},\ldots,t_n^{b_{mn}})])=[(1,t_1^{b_{21}} \cdots t_n^{b_{2n}},\ldots,t_1^{b_{m1}} \cdots t_n^{b_{mn}})]$.
\end{center}

It is immediate that $\phi$ is an epimorphism between multiplicative groups. Thus
$$
 |X|=\frac{|Y_1 \times \cdots \times Y_n|}{|{\rm ker} \, \phi|}=\frac{1}{{|{\rm ker} \, \phi|}} \prod_{i=1}^n |Y_i|.
$$

Let $\beta$ a generator of $(K^*,\cdot)$. Therefore

\begin{center}
${\rm ker} \, \phi=\{([(1,\beta^{i_1b_{21}},\ldots,\beta^{i_1b_{m1}})],\ldots,[(\beta^{i_nb_{2n}},\ldots,\beta^{i_nb_{mn}})]) \in Y_1 \times \cdots \times Y_n : [(1,\beta^{i_1b_{21}+\cdots+i_nb_{2n}},\ldots,\beta^{i_{1}b_{m1}+\cdots+i_nb_{mn}})]=[(1,1,\ldots,1)]\}$.
\end{center}

In this case $\beta^{i_1b_{21}+\cdots+i_nb_{2n}}=1,\ldots,\beta^{i_1b_{m1}+\cdots+i_nb_{mn}}=1$. These equalities imply the system of congruences (\ref{system}). Then there is a bijection between ${\rm ker} \, \phi$ and the set of $n-$tuples $(i_1,\ldots,i_n)$ such that $1 \leq i_j \leq |Y_j|$ for all $j=1,\ldots,n$ and satisfy (\ref{system}).

The Eq. (\ref{length}) follows immediately from last results. 
\end{proof}

We define the projective torus of dimension $m-1$ as
\begin{equation}
\mathbb{T}_{m-1}=\{[(c_1,\ldots,c_m)] \in {\mathbb{P}}^{m-1} : c_i \in K^* \, \, \mbox{for all} \, i\}.
\end{equation}

Obviously, $X \subseteq \mathbb{T}_{m-1}$. The following corollary is an easy consequence of Theorem \ref{theorem-length}. It gives the conditions under which last inclusion becomes an equality.

\begin{corollary}
If $n= m$ then $X$ is the projective torus of dimension $m-1$ if and only if $|M|=1$ and $(q-1,b_{2j},\ldots,b_{mj})=1$ for all $j=1,\ldots,n$.
\end{corollary}

On the other hand if we consider the case where the monomials that parameterize the toric set $X$ are all of them of the same degree then we obtain another corollary.

\begin{corollary} \label{cor2}
If the sum of the elements of each column of the matrix $A$ defined in (\ref{matrix}) is a constant or, equivalently, the monomials that parameterize the toric set $X$ are all of them of the same degree, then $|X| \leq (q-1)^{n-1}$.
\end{corollary}

\begin{proof}
Let $\sum_{j=1}^n a_{ij}=\alpha$ (a positive integer) for all $i=1,\ldots,m$. We note that $|Y_i|\leq q-1$ for all $i=1,\ldots,n$. Moreover

\begin{center}
$\sum_{j=1}^n b_{ij}=\sum_{j=1}^n a_{ij}-\sum_{j=1}^n a_{1j}=\alpha-\alpha=0$,
\end{center}

\noindent and then $(1,\ldots,1) \in M$. Let $\gamma={\rm min} \, \{|Y_i|:i=1,\ldots,n\}$. Therefore $(j,\ldots,j) \in M$ for all $1 \leq j \leq \gamma$ and it implies that $|M| \geq \gamma$. Thus

\begin{center}
$|X|=\frac{1}{|M|} \prod_{i=1}^n |Y_i| \leq \frac{\gamma \, (q-1)^{n-1}}{\gamma}=(q-1)^{n-1}$
\end{center}

\noindent and the claim follows.
\end{proof}

\begin{remark} \label{remark1}
If $G$ is a graph and $X$ is the algebraic toric set associated to the incidence matrix of $G$, then the sum of the elements of each column of this matrix is $\alpha=2$ and the result of the last corollary follows. Actually in this situation $|Y_i|=q-1$ for all $i=1,\ldots,n$ and then we get that in any graph $|X|=\frac{(q-1)^n}{|M|}$. Moreover in \cite[Corollary 3.8]{algcodes} it was found the exact value of $|X|$ if $G$ is a connected graph. By using this result we obtain that
$$
|M| = \left\{\begin{array}{lll}
(q-1)^2 & {\rm if} & \, G \,{\mbox{is bipartite}} \\
q-1 & {\rm if} & \, G \, {\mbox{is non-bipartite}}.
\end{array}\right.
$$
\end{remark}

On the other hand if we consider disconnected graphs then we obtain the following result.

\begin{corollary}
Let $(q,2)=1$ and $G$ be a disconnected graph with $n$ vertices and $m$ edges. If $X$ is the algebraic toric set associated to the incidence matrix of the graph $G$, then

\begin{center}
$|X|<(q-1)^{n-1}$.
\end{center}
\end{corollary}

\begin{proof}
Let $E=\{a_1,\ldots,a_m\}$ be the edge set of $G$, where we consider that $a_1,\ldots,a_m$ are the columns of the incidence matrix of $G$. There is no loss of generality if we consider that $\{a_1,\ldots,a_{s_1}\}$, with ${s_1}<m$, corresponds to a connected component of $G$. In the same way let $V=\{v_1,\ldots,v_n\}$ be the set of vertices of $G$ where we can suppose that $\{v_1,\ldots,v_{s_2}\}$ is the set of vertices of the connected component mentioned above with $s_2 < n$. In the proof of Corollary (\ref{cor2}) and Remark \ref{remark1} it was showed that $(j,j,\ldots,j) \in M$ for all $j=1,\ldots,q-1$ and then $|M| \geq q-1$. In this case the element

\begin{center}
$\overbrace{(1,\ldots,1}^{{s_2}-{\rm entries}},\overbrace{(q+1)/2,\ldots,(q+1)/2)}^{(n-{s_2})-{\rm entries}}$
\end{center}

\noindent belongs to $M$ and thus $|M| >q-1$. Therefore

\begin{center}
$|X|=\frac{(q-1)^n}{|M|}<(q-1)^{n-1}$,
\end{center}

\noindent and the claim follows.
\end{proof}

\subsection{Dimension}

\noindent In the following theorem we give the dimension of the projective parameterized codes arising from the algebraic toric set $X$ in terms of the dimension of the projective parameterized codes arising from the projective torus $\mathbb{T}_{m-1}$, which is well known (see \cite{GRH}).

\begin{theorem} \label{dimension1}
The dimension of the projective parameterized codes of order $d$, $C_X(d)$, is given by
\begin{equation} \label{dimension}
H_X(d)=H_{\mathbb{T}_m-1}(d)- \overline{H}(d) 
\end{equation}

\noindent for all $d \geq 0$ and where $\overline{H}$ is the Hilbert function of $I_X/I_{\mathbb{T}_{m-1}}$, i.e., $$\overline{H}(d)={\rm dim}_K \, I_X(d)/I_{\mathbb{T}_{m-1}}(d).$$
\end{theorem}
\begin{proof}
\noindent We know that $X \subseteq {\mathbb{T}}_{m-1}$ and then $I_{\mathbb{T}_{m-1}} \subseteq I_X$. Let $\psi$ be the following linear transformation.
\begin{eqnarray}
\psi : S_d/I_{\mathbb{T}_{m-1}}(d) \rightarrow S_d/I_X(d), \nonumber\\
f+I_{\mathbb{T}_{m-1}}(d) \rightarrow f+I_X(d).
\end{eqnarray}

This a well defined function and in fact it is a surjective linear map. Moreover ${\rm ker} \, \psi=I_X(d)/I_{\mathbb{T}_{m-1}}(d)$. Thus
$$
{\rm dim}_K S_d/I_{\mathbb{T}_{m-1}}(d)=
$$
$$
{\rm dim}_K S_d/I_X(d)+ {\rm dim}_K I_X(d)/I_{\mathbb{T}_{m-1}}(d),
$$
\noindent and the equality (\ref{dimension}) follows immediately.
\end{proof}

For the following corollary we will use $r_X, r_{\mathbb{T}_{m-1}}$ and $r_{\overline{H}}$ as the regularity indexes of $S/I_X$, $S/I_{\mathbb{T}_{m-1}}$ and $I_X/I_{\mathbb{T}_{m-1}}$, respectively.

\begin{corollary} \label{regularity}
$r_{\mathbb{T}_{m-1}}={\rm max} \, \{r_X, r_{\overline{H}}\}$.
\end{corollary}

\begin{proof}
Let
\begin{center}
$\theta : I_X(d)/I_{\mathbb{T}_{m-1}}(d) \rightarrow I_X(d+1)/I_{\mathbb{T}_{m-1}}(d+1)$,
\end{center}
\begin{center}
$\theta(f+I_{\mathbb{T}_{m-1}}(d))=X_1 f+I_{\mathbb{T}_{m-1}}(d+1)$. 
\end{center}
It is easy to see that $\theta$ is a well defined map, moreover it is a li\-near transformation. If $f+I_{\mathbb{T}_{m-1}}(d) \in {\rm ker} \, \theta$ then $X_1f \in I_{\mathbb{T}_{m-1}}(d+1)$. Let $[P]=[(1,t_1^{b_{21}}\cdots t_n^{b_{2n}},\ldots,t_1^{b_{m1}}\cdots t_n^{b_{mn}})] \in X$. Thus $(X_1f)(P)=0$ and then $f(P)=0$. Therefore $f \in I_{\mathbb{T}_{m-1}}(d)$ and ${\rm ker} \, \theta=I_{\mathbb{T}_{m-1}}(d)$. It implies that $\overline{H}(d) \leq \overline{H}(d+1)$ for all $d \geq 0$. By the last inequality and Eq. (\ref{dimension}) the claim follows. 
\end{proof}

\begin{remark}
By the Corollary \ref{regularity} we obtain that $r_X \leq r_{\mathbb{T}_{m-1}}$. But in \cite{GRH} it was proved that $r_{\mathbb{T}_{m-1}}=(m-1)(q-2)$. Therefore
\begin{equation}
r_X \leq (m-1)(q-2).
\end{equation} 

As we observed in section \ref{sec2}, if $\, d \geq (m-1)(q-2)$ then $H_X(d)=|X|$ and thus $C_X(d)=K^{|X|}$. Therefore from now on we will use $d < (m-1)(q-2)$.
\end{remark}

\subsection{Minimum distance}
The minimum distance has been computed in several cases associated to evaluation codes. In particular in \cite{ci-codes} it was computed when we consider projective parameterized codes arising from the projective torus. Moreover in \cite{tohaneanu} some lower bounds on the minimum distance were found coming from syzigies. In this section we are going to find an upper bound for the minimum distance of any projective parameterized code and we will find a lower bound for this kind of codes when the sum of the elements of each column of the matrix (\ref{matrix}) becomes a constant.
We consider that $X \subset \mathbb{T}_{m-1}$ because the case $X=\mathbb{T}_{m-1}$ is well known (see \cite{GRH}). Let $Y:=\mathbb{T}_{m-1} \setminus X$ and $\delta_X(d)$, $\delta_Y(d)$ and $\delta_{\mathbb{T}_{m-1}}(d)$ be the minimum distances of the pa\-ra\-me\-te\-ri\-zed codes of order $d$, $C_X(d)$, $C_Y(d)$ and $C_{\mathbb{T}_{m-1}}(d)$, respectively. The following theorem relates them.

\begin{theorem}
Let $0 \leq d < (m-1)(q-2)$. Then
\begin{equation} \label{inequalities}
\delta_X(d) \leq \delta_{\mathbb{T}_{m-1}}(d)-\delta_Y(d).
\end{equation}
\end{theorem}

\begin{proof}
Let $X=\{[P_1],\ldots,[P_{|X|}]\}$. We can write $\mathbb{T}_{m-1}=\{[P_1],\ldots,[P_{|X|}],[Q_1],\ldots,[Q_{|Y|}]\}$, where of course $Y=\{[Q_1],\ldots,[Q_{|Y|}]\}$. If $$\Lambda=\left(\frac{f(P_1)}{X_1^d(P_1)},\ldots,\frac{f(P_{|X|})}{X_1^d(P_{|X|})}, \frac{f(Q_1)}{X_1^d(Q_1)},\ldots, \frac{f(Q_{|Y|})}{X_1^d(Q_{|Y|})}\right) \in C_{\mathbb{T}_{m-1}}(d)$$
with $w(\Lambda)=\delta_{\mathbb{T}_{m-1}}(d)$. We use $w(\Lambda)$ to mean the Hamming weight of the codeword $\Lambda$, then 
$$
\Lambda_1:=\left(\frac{f(P_1)}{X_1^d(P_1)},\ldots,\frac{f(P_{|X|})}{X_1^d(P_{|X|})}\right) \in C_X(d)
\,\,\,\,\mbox{and}\,\,\,\,
\Lambda_2:=\left(\frac{f(Q_1)}{X_1^d(Q_1)},\ldots,\frac{f(Q_{|Y|})}{X_1^d(Q_{|Y|})}\right) \in C_Y(d).$$

Moreover
\begin{equation} \label{inequalities2}
\delta_{\mathbb{T}_{m-1}}(d)=w(\Lambda)=w(\Lambda_1)+w(\Lambda_2) \geq \delta_X(d)+\delta_{Y}(d).
\end{equation}

Therefore the inequality (\ref{inequalities}) follows from (\ref{inequalities2}).
\end{proof}

\begin{remark}
From the inequality (\ref{inequalities}) we obtain that $\delta_X(d) \leq \delta_{\mathbb{T}_m-1}(d)-1$ for all $0 \leq d < (m-1)(q-2)$. But $\delta_{\mathbb{T}_{m-1}}(d)$ was computed in \cite{ci-codes}. Thus in this case
\begin{equation}
\delta_X(d) \leq (q-1)^{m-(k+2)}(q-1-\ell)-1, 
\end{equation}

\noindent where $k$ and $\ell$ are the unique integers such that $k\geq 0$,
$1\leq \ell\leq q-2$ and $d=k(q-2)+\ell$. 
\end{remark}

From now on we will consider the case worked in section \ref{length1}, where the sum of the elements of each column of the matrix $A$ defined in (\ref{matrix}) is a constant, i.e., $\sum_{j=1}^n a_{ij}=\alpha$ (a positive integer) for all $i=1,\ldots,m$. The following map will help us to find a lower bound for the minimum distance of the corresponding projective parameterized codes.
$$
\mu: {\mathbb{T}_{n-1}} \rightarrow X,$$
$$
\left[(t_1,\ldots,t_n)\right] \rightarrow \left[\left(t_1^{a_{11}}\cdots t_n^{a_{1n}},\ldots,t_1^{a_{m1}}\cdots t_n^{a_{mn}}\right)\right].$$ 

$\mu$ is a well defined map and in fact it is an epimorphism of multiplicative groups. Let $N:={\rm ker} \, \mu$. Thus $|N|=\frac{|\mathbb{T}_{n-1}|}{|X|}=\frac{(q-1)^{n-1}}{|X|}$. Moreover $\mathbb{T}_{n-1}=\cup_{i=1}^{|X|} N\cdot [Q_i]$ (disjoint union of the corresponding cosets) for some $[Q_i] \in \mathbb{T}_{n-1}$. Let $[P_i]=\mu([Q_i])$ for all $i=1,\ldots,|X|$ and $N=\{[R_1],\ldots,[R_{|N|}]\}$. Thus $X=\{[P_1],\ldots,[P_{|X|}]\}$ and 

\begin{center}
$\mathbb{T}_{n-1}=\{[R_1Q_1],\ldots,[R_{|N|}Q_1],\ldots,[R_1Q_{|X|}],\ldots,[R_{|N|}Q_{|X|}]\}$.
\end{center}

As in the introduction let $L=K[Z_1,\ldots,Z_n]$. We define another map that will be useful later on.
$$
\tau:S_d \rightarrow L_{\alpha d},$$ $$f(X_1,\ldots,X_m) \rightarrow f(Z_1^{a_{11}}\cdots Z_n^{a_{1n}},\ldots,Z_1^{a_{m1}}\cdots Z_n^{a_{mn}}).
$$

$\tau$ is a linear map between the vector spaces $S_d$ and $L_{\alpha d}$. Now we are able to prove the following theorem. In this result we are going to find a lower bound for the minimum distance of the corresponding projective parameterized codes.

\begin{theorem} \label{theoremlowerbound}
If the sum of the elements of each column of the matrix $A$ defined in (\ref{matrix}) is a constant $\alpha$, then
\begin{equation} \label{lowerbound}
\delta_X(d) \geq \frac{|X| \cdot \delta_{\mathbb{T}_{n-1}}(\alpha d)}{(q-1)^{n-1}},
\end{equation}

\noindent where $\delta_{\mathbb{T}_{n-1}}(\alpha d)$ is the minimum distance of the parametererized code of order $\alpha d$ arising from the projective torus $\mathbb{T}_{n-1}$ and $\delta_X(d)$ is the minimum distance of the projective parameterized code associated to the toric set $X$ defined in Eq. (\ref{toricset}).
\end{theorem}

\begin{proof}
Let 
$$\Gamma=\left(\frac{f(P_1)}{X_1^d(P_1)},\ldots,\frac{f(P_{|X|})}{X_1^d(P_{|X|})}\right) \in C_X(d).$$ 

We choose $\Gamma$ in such a way that $w(\Gamma)=\delta_X(d)$. On the other hand let
$$\Omega=\left(\frac{\tau(f)(R_1Q_1)}{Z_1^{\alpha d}(R_1Q_1)},\ldots,\frac{\tau(f)(R_{|N|}Q_1)}{Z_1^{\alpha d}(R_{|N|}Q_1)},\ldots,\frac{\tau(f)(R_1Q_{|X|})}{Z_1^{\alpha d}(R_1Q_{|X|})},\ldots,\frac{\tau(f)(R_{|N|}Q_{|X|})}{Z_1^{\alpha d}(R_{|N|}Q_{|X|})}\right).$$ 

We have that $\Omega \in C_{\mathbb{T}_{n-1}}(\alpha d)$ and if $f(P_i) \neq 0$ for some $[P_i] \in X$, then due to the fact that $\mu ([R_jQ_i])=[P_i]$, we obtain that $\tau(f)(R_j Q_i)=f(P_i)\neq 0$ for all $j=1,\ldots,n$. Thus $w(\Omega)=|N|\cdot w(\Gamma)=|N|\cdot \delta_X(d)$ and thererefore $\delta_{\mathbb{T}_{n-1}}(\alpha d) \leq w(\Omega)=|N|\cdot \delta_X(d)$. Then
\begin{equation} \label{inequality2}
\delta_X(d) \geq \frac{\delta_{\mathbb{T}_{n-1}}(\alpha d)}{|N|}.
\end{equation}

The inequality (\ref{lowerbound}) follows from (\ref{inequality2}) and the fact that $|N|=\frac{(q-1)^{n-1}}{|X|}$.
\end{proof}

\vspace{0.3cm}
If $X$ is the algebraic toric set arising from the incidence matrix of any graph then $\alpha=2$ and we can apply Theorem \ref{theoremlowerbound}. Moreover if we have a connected graph, by using \cite[Corollary 3.8]{algcodes} we obtain the following general result.

\begin{corollary} \label{mdist}
Let $X$ be the algebraic toric set arising from the incidence matrix of any connected graph $G$. Then
\begin{equation} \label{mdist2}
\delta_X(d) \geq \left\{\begin{array}{lll}
\frac{\delta_{\mathbb{T}_{n-1}}(2d)}{q-1} & {\rm if} & \, G \,{\mbox{is bipartite}} \\
\delta_{\mathbb{T}_{n-1}}(2d) & {\rm if} & \, G \, {\mbox{is non-bipartite}}.
\end{array}\right.
\end{equation}
\end{corollary}

\begin{corollary}
If the sum of the elements of each column of the matrix $A$ defined in (\ref{matrix}) is a constant $\alpha$, then
\begin{equation} \label{index3}
r_X \geq \frac{|X|(q-2)(n-1)}{\alpha (q-1)^{n-1}},
\end{equation}

\noindent where $r_X$ is the regularity index of $S/I_X$.

Moreover if $G$ is a connected graph and $X$ is the algebraic toric set arising from its incidence matrix, then
\begin{equation} \label{index4}
r_X \geq \left\{\begin{array}{llll}
\frac{(q-2)(n-1)}{2(q-1)} & {\rm if} & \, G \,{\mbox{is bipartite}} \\
& & \\
\frac{(q-2)(n-1)}{2} & {\rm if} & \, G \, {\mbox{is non-bipartite}}.
\end{array}\right.
\end{equation}
\end{corollary}

\begin{proof}
The claim follows directly because of (\ref{lowerbound}), (\ref{mdist2}), and the fact that the regularity index corresponding to the torus $\mathbb{T}_{n-1}$ is exactly $(q-2)(n-1)$.
\end{proof}
In the first example of the following section we will realize that this lower bound is attained in some cases.
\section{Examples}
In this section we will give three different examples. In the first example we will consider a particular connected non-bipartite graph and we will compute the main characteristics of the corresponding projective parameterized codes arising from the incidence matrix of that graph. In the second example we will define clutters as particular cases of hypergraphs and a specific example of projective parameterized codes arising from uniform clutters will be given. Finally in the third example we will compute the main parameters of the projective parameterized codes associated to a matrix that does not represent a clutter and then it does not represent a graph. In these examples we will use the notation appeared in the previous sections and we will use Macaulay2 \cite{mac2} for the main computations. Also we will use $\delta'_d$ to represent the lower bound showed in (\ref{lowerbound}) and $b_d$ will represent the Singleton bound, i.e.,

\begin{center}
$\delta'_d=\frac{|X| \cdot \delta_{\mathbb{T}_{n-1}}(\alpha d)}{(q-1)^{n-1}} \,$ and $\, b_d=|X|-H_X(d)+1$.
\end{center}

In the following examples we will take $\delta'_d=1$ in the cases where $\delta'_d \leq 1$. 

\begin{figure}[!t]
\centering
\begin{picture}(260,120)
 \put(100,80){\circle*{6}}
 \put(140,80){\circle*{6}}
 \put(100,0){\circle*{6}}
 \put(140,0){\circle*{6}}
 \put(120,40){\circle*{6}}
 \curve(100,80,140,0)
 \curve(140,80,100,0)
 \curve(100,80,140,80)
 \curve(100,0,140,0)
\put(125,40){$v_1$}
\put(95,90){$v_2$}
\put(135,90){$v_3$}
\put(95,-10){$v_4$}
\put(140,-10){$v_5$}
\put(95,60){$a_1$}
\put(115,85){$a_2$}
\put(135,60){$a_3$}
\put(95,20){$a_4$}
\put(135,20){$a_6$}
\put(120,-8){$a_5$}
\end{picture}
\vspace{0.1cm}
\caption{A connected non-bipartite graph with two cycles of length $3$.}
\label{fig1}
\end{figure}
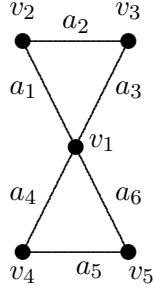

\subsection{Example 1}
Let $G$ be the graph given in Fig. \ref{fig1} where $V(G)=\{v_1,v_2,v_3,v_4,v_5\}$ is its vertex set and its edge set is given by $E(G)=\{a_1,a_2,a_3,a_4,a_5,a_6\}$. The incidence matrix of $G$ is the $5 \times 6$ matrix given by
\begin{equation}
\label{matrix2}
A=\begin{pmatrix}
1 & 0 & 1 & 1 & 0 & 1 \\
1 & 1 & 0 & 0 & 0 & 0 \\
0 & 1 & 1 & 0 & 0 & 0 \\
0 & 0 & 0 & 1 & 1 & 0 \\
0 & 0 & 0 & 0 & 1 & 1
\end{pmatrix}.
\end{equation}

Let $K=\mathbb{F}_7$ be a finite field with $7$ elements. The toric set arising from the matrix (\ref{matrix2}) (or associated to the graph $G$ showed in Fig. \ref{fig1}) is given by

\begin{center}
$X=\{[(t_1t_2,t_2t_3,t_1t_3,t_1t_4,t_4t_5,t_1t_5)] \in \mathbb{P}^5: t_i \in K^*\}$.
\end{center}

In this case we have five subsets $Y_i$ with $|Y_i|=6$ for all $i=1,\ldots,5$. The corresponding subset $M$ is

\begin{center}
$M=\{(i,i,i,i,i):i=1,\ldots,6\}$,
\end{center}

\noindent and therefore, by using Theorem \ref{theorem-length},

\begin{center}
$|X|=\frac{1}{|M|} \prod_{i=1}^5 |Y_i|=1296$.
\end{center}

We notice that $\delta'_d=\delta_{\mathbb{T}_4}(2d)$ because of Corollary \ref{mdist}. By using Macaulay2 we compute the following values.

\begin{center}
\begin{tabular}{||c||c||c||c||c||c||} \hline
  $d$ & 1 & 2 & 3 & 4 & 5 \\ \hline
  $H_X(d)$ & 6 & 21 & 55 & 120 & 231\\ \hline
  $H_{\mathbb{T}_5}(d)$ & 6 & 21 & 56 & 126 & 252\\ \hline
  $\overline{H}(d)$ & 0 & 0 & 1 & 6 & 21\\ \hline
  $\delta'_d$ & 864 & 432 & 180 & 108 & 36\\ \hline
  $b_d$ & 1291 & 1276 & 1242 & 1177 & 1066 \\ \hline
\end{tabular}
\end{center}

\begin{center}
\begin{tabular}{||c||c||c||c||c||c||} \hline
  $d$ & 6 & 7 & 8 & 9 & 10 \\ \hline
  $H_X(d)$ & 401 & 627 & 885 & 1130 & 1296\\ \hline
  $H_{\mathbb{T}_5}(d)$ & 457 & 762 & 1182 & 1722 & 2373\\ \hline
  $\overline{H}(d)$ & 56 & 135 & 297 & 592 & 1077\\ \hline
  $\delta'_d$ & 24 & 12 & 5 & 3 & 1\\ \hline
  $b_d$ & 896 & 670 & 412 & 167 & 1 \\ \hline
\end{tabular}
\end{center}

Moreover in this case the regularity index is $r_X=10=\frac{(q-2)(n-1)}{2}$ and it shows that the lower bound given in (\ref{index4}) works very well. This lower bound is attained in this particular case.

\subsection{Example 2}
In this example we continue using the notation used in the introduction.

A clutter $\mathcal{C}$ is a family $E$ of subsets of a finite ground set $Z=\{Z_1,\ldots,Z_n\}$ such that if $h_1, h_2 \in E$, then $h_1 \not\subset h_2$. The ground set $Z$ is called the vertex set of $\mathcal{C}$ and $E$ is called the edge set of $\mathcal{C}$ and they are denoted by $V_{\mathcal{C}}$ and $E_{\mathcal{C}}$ respectively.

Clutters are special hypergraphs and are sometimes called Sperner families in the literature. One example of a Clutter is a graph with the vertices and edges defined in the usual way for graphs.

Let $\mathcal{C}$ be a clutter with vertex set $V_{\mathcal{C}}=\{Z_1,\ldots,Z_n\}$ and let $h$ be an edge of $\mathcal{C}$. The characteristic vector of $h$ is the vector $a=\sum_{Z_i \in h} e_i$ where $e_i$ is the $ith$ unit vector in $\mathbb{R}^n$. Throughout this example we assume that $a_1,\ldots,a_m$ is the set of all characteristic vectors of the edges of $\mathcal{C}$. In this case the matrix (\ref{matrix}) is known as the incidence matrix of the clutter $\mathcal{C}$ and the set $X$ defined in (\ref{toricset}) is the toric set associated to the clutter $\mathcal{C}$. The clutter $\mathcal{C}$ is called uniform if the sum of the elements of the columns of its incidence matrix is a constant.

We realize that in any clutter, like in graphs, $|X|=\frac{(q-1)^n}{|M|}$ because $|Y_i|=q-1$ for all $i=1,\ldots,n$.

Let $K=\mathbb{F}_9$ be a finite field with $9$ elements and $X$ be the toric set associated to the uniform clutter ($\alpha=3$) whose incidence matrix is the $6 \times 6$ matrix given by

\begin{equation}
\label{matrix3}
A=\begin{pmatrix}
1 & 0 & 0 & 0 & 1 & 1 \\
1 & 1 & 0 & 0 & 0 & 1 \\
1 & 1 & 1 & 0 & 0 & 0 \\
0 & 1 & 1 & 1 & 0 & 0 \\
0 & 0 & 1 & 1 & 1 & 0 \\
0 & 0 & 0 & 1 & 1 & 1
\end{pmatrix}
\end{equation}

The toric set $X$ associated to (\ref{matrix3}) becomes

\begin{center}
$X=\{[(t_1t_2t_3,t_2t_3t_4,t_3t_4t_5,t_4t_5t_6,t_1t_5t_6,t_1t_2t_6)] \in \mathbb{P}^5: t_i \in K^*\}$.
\end{center}

In this case we have six subsets $Y_i$ with $|Y_i|=8$ for $i=1,\ldots,6$. The corresponding set $M$ used in Theorem \ref{theorem-length} has $512$ elements and therefore, by Eq. (\ref{length}),

\begin{center}
$|X|=\frac{1}{|M|} \prod_{i=1}^6 |Y_i|=512$.
\end{center}

In the same way that in the last example we obtain, by using Macaulay2, the following values.

\begin{center}
\begin{tabular}{||c||c||c||c||c||c||c||} \hline
  $d$ & 1 & 2 & 3 & 4 & 5 & 6\\ \hline
  $H_X(d)$ & 6 & 19 & 44 & 85 & 146 & 231\\ \hline
  $H_{\mathbb{T}_5}(d)$ & 6 & 21 & 56 & 126 & 252 & 462\\ \hline
  $\overline{H}(d)$ & 0 & 2 & 12 & 41 & 106 & 231\\ \hline
  $\delta'_d$ & 320 & 128 & 48 & 24 & 7 & 4\\ \hline
  $b_d$ & 507 & 494 & 469 & 428 & 367 & 282 \\ \hline
\end{tabular}
\end{center}

\begin{center}
\begin{tabular}{||c||c||c||c||c||c||c||} \hline
  $d$ & 7 & 8 & 9 & 10 & 11 & 12 \\ \hline
  $H_X(d)$ & 344 & 442 & 492 & 510 & 512 & 512\\ \hline
  $H_{\mathbb{T}_5}(d)$ & 792 & 1282 & 1972 & 2898 & 4088 & 5558\\ \hline
  $\overline{H}(d)$ & 448 & 840 & 1480 & 2388 & 3576 & 5046\\ \hline
  $\delta'_d$ & 1 & 1 & 1 & 1 & 1 & 1\\ \hline
  $b_d$ & 169 & 71 & 21 & 3 & 1 & 1\\ \hline
\end{tabular}
\end{center}

It is immediate from the last table that $r_X=11$.

\subsection{Example 3}
In this example we will give the main characteristics of the projective parameterized codes arising from a matrix that does not represent a clutter.

Let $K=\mathbb{F}_{11}$ be a finite field with $11$ elements and $X$ be the toric set associated to the $3 \times 4$ matrix given by
\begin{equation}
\label{matrix5}
A=\begin{pmatrix}
3 & 1 & 0 & 1 \\
0 & 4 & 2 & 2 \\
3 & 1 & 4 & 3 
\end{pmatrix}
\end{equation}

In this case $\alpha=6$ and the set $X$ becomes

\begin{center}
$X=\{[(t_1^3t_3^3,t_1t_2^4t_3,t_2^2t_3^4,t_1t_2^2t_3^3)] \in \mathbb{P}^3:t_i \in K^*\}$.
\end{center}

We have three subsets $Y_i$ with $|Y_1|=|Y_3|=10$ and $|Y_2|=5$. The corresponding subset $M$ has $10$ elements and then, by Theorem \ref{theorem-length},

\begin{center}
$|X|=\frac{1}{|M|} \prod_{i=1}^3 |Y_i|=50$.
\end{center}

By using Macaulay2 we obtain the following values.

\begin{center}
\begin{tabular}{||c||c||c||c||c||c||c||} \hline
  $d$ & 1 & 2 & 3 & 4 & 5 & 6\\ \hline
  $H_X(d)$ & 4 & 10 & 20 & 32 & 44 & 50\\ \hline
  $H_{\mathbb{T}_3}(d)$ & 4 & 10 & 20 & 35 & 56 & 84\\ \hline
  $\overline{H}(d)$ & 0 & 0 & 0 & 3 & 12 & 34\\ \hline
  $\delta'_d$ & 20 & 3 & 1 & 1 & 1 & 1\\ \hline
  $b_d$ & 47 & 41 & 31 & 19 & 7 & 1 \\ \hline
\end{tabular}
\end{center}

We conclude that $r_X=6$.

\end{document}